\documentclass[11pt]{article}
\usepackage[margin=2.5cm]{geometry}
\usepackage{amsmath,amssymb,amsthm}
\usepackage{booktabs}

\usepackage{hyperref}

\newcommand{\E}{\mathbb{E}}
\newcommand{\R}{\mathbb{R}}
\newcommand{\ind}{\mathbf{1}}
\newcommand{\dd}{\mathrm{d}}
\newcommand{\pp}{\partial}

\newtheorem{theorem}{Theorem}
\newtheorem{corollary}{Corollary}
\newtheorem{remark}{Remark}

\title{Unbiased Monte Carlo Greeks for Discontinuous Payoffs}
\author{Evgeny Lakshtanov\\[4pt]
{\normalsize CIDMA, University of Aveiro, Portugal}\\[2pt]
{\small \texttt{lakshtanov@ua.pt}}}
\date{\today}

\begin{document}
\maketitle

\begin{abstract}
Pathwise differentiation of Monte Carlo estimators fails at payoff discontinuities, producing zero or biased sensitivities for barriers, autocallables, and digital options. The industry workaround --- smoothing the indicator functions --- introduces bias and requires per-product calibration. We derive a correction formula that restores unbiased Greeks without smoothing. For a payoff $F(Z,\theta)$ that is piecewise smooth with discontinuities on surfaces $\{g_i = 0\}$, we show that the sensitivity decomposes into a pathwise term (computed by standard AAD) plus a sum of boundary corrections, each involving the payoff jump, the Gaussian density at the boundary, and the sensitivity of the boundary to the parameter. The correction is computed by Newton root-finding in the normal-random space, with the jump evaluated by two forward replays of the pricing kernel. The implementation uses AADC (\texttt{pip install aadc}), whose tape replay and automatic discontinuity tracking make the method fully automatic --- the quant writes standard pricing code, and the correction driver identifies and handles all discontinuities. We prove the formula for arbitrary compositions of smooth functions and indicator functions (not just outer products), covering real autocallable payoff structures with recursive alive/dead logic. Benchmarks on QuantLib models (GBM, Heston, Hull-White) show all Greeks within 0.1--4\% of analytic or bump-and-revalue references.
\end{abstract}

\section{Introduction}

Consider a derivative whose price is given by
\begin{equation}\label{eq:price}
V(\theta) = \E\bigl[F(Z, \theta)\bigr], \qquad Z \sim \mathcal{N}(0, I_d),
\end{equation}
where $Z = (Z_1, \ldots, Z_d)$ are independent standard normal random variables driving the simulation (e.g., $d = \text{assets} \times \text{time steps}$), $\theta$ is a model parameter (spot, volatility, rate, etc.), and $F$ is the discounted payoff.

The \emph{pathwise method} computes the sensitivity $\pp V / \pp \theta$ by differentiating under the expectation:
\[
\frac{\pp V}{\pp \theta} = \E\!\left[\frac{\pp F}{\pp \theta}\right].
\]
This works when $F$ is smooth in $\theta$, and is efficiently computed by adjoint algorithmic differentiation (AAD) at the cost of a single reverse pass.

However, many structured products contain indicator functions:
\[
F(Z, \theta) = P(Z, \theta) \cdot \prod_{i=1}^K \ind\{g_i(Z, \theta) > 0\},
\]
or more generally, $F$ is a piecewise smooth function with discontinuities on the surfaces $\Gamma_i = \{g_i = 0\}$. Naively differentiating produces $\delta$-functions:
\[
\frac{\pp}{\pp\theta}\bigl(P \cdot \ind\{g > 0\}\bigr) = \frac{\pp P}{\pp\theta} \cdot \ind\{g > 0\} + P \cdot \delta(g) \cdot \frac{\pp g}{\pp\theta},
\]
and the $\delta(g)$ term is zero at almost every sample point. AAD returns zero sensitivity for paths near the boundary.

The standard workaround is \emph{smoothing}: replacing $\ind\{g > 0\}$ with $\Phi(g/\varepsilon)$. This introduces bias of order $\varepsilon$, requires calibration of $\varepsilon$ per product and per Greek, and creates model risk.

Bump-and-revalue also suffers from discontinuities: the variance of the finite-difference estimator $\Delta \approx (V(\theta+h) - V(\theta-h))/(2h)$ grows as $O(1/h)$ for discontinuous payoffs, compared to $O(1)$ for smooth payoffs.

In this paper we derive a correction formula that replaces the $\delta$-function with a computable boundary term. The correction is unbiased, requires no parameter tuning, and works for any model (GBM, Heston, Hull-White, hybrid).

\paragraph{Origin of the idea.} In 2021, while working at Matlogica, Dmitry Goloubentsev posed the question: can we reduce the $d$-dimensional expectation to a $(d-1)$-dimensional expectation of a one-dimensional integral that is smoother than the original? At the time we did not obtain interesting results. The breakthrough came later: the 1D integral does not need to be evaluated explicitly --- the Leibniz rule extracts the boundary contribution directly.

\section{Main results}

\subsection{Single indicator}

\begin{theorem}\label{thm:single}
Let $P, g: \R^d \times \R \to \R$ be smooth, with $\nabla_{\!Z} g \neq 0$ on $\Gamma = \{g = 0\}$. Define $F(Z,\theta) = P(Z,\theta) \cdot \ind\{g(Z,\theta) > 0\}$. Let $v \in \R^d$ be a fixed unit vector with $v \cdot \nabla_{\!Z} g > 0$ on $\Gamma$. Write $Z = u\,v + w$ with $u = v \cdot Z \in \R$, $w \perp v$.

For each $w \in \R^{d-1}$, let $u^*(w, \theta)$ be the solution of $g(u^*v + w, \theta) = 0$, and set $Z^*(w) = u^*v + w$.

Then for $V(\theta) = \E_Z\bigl[F(Z,\theta)\bigr]$, $Z \sim \mathcal{N}(0, I_d)$:
\begin{equation}\label{eq:correction}
\frac{\pp V}{\pp \theta} = \E_Z\!\left[\frac{\pp F}{\pp \theta}\right]
+ \E_w\!\left[\Delta F \cdot \frac{\varphi(u^*)}{\pp_u g} \cdot \frac{\pp g}{\pp \theta}\bigg|_{Z^*}\right],
\end{equation}
where the first expectation is over $Z \sim \mathcal{N}(0, I_d)$ (pathwise, with $\pp F/\pp\theta = 0$ where $g < 0$), the second over $w \sim \mathcal{N}(0, I_{d-1})$, $\Delta F = \lim_{\varepsilon \to 0^+} [F(Z^* + \varepsilon v) - F(Z^* - \varepsilon v)]$ is the one-sided payoff jump across the boundary, $\pp_u g = v \cdot \nabla_{\!Z} g$, and $\varphi(x) = (2\pi)^{-1/2} e^{-x^2/2}$.
\end{theorem}

\begin{proof}
\textbf{Step 1 (change of variables).} Let $v \in \R^d$ be any fixed unit vector such that $v \cdot \nabla_{\!Z} g \neq 0$ on $\Gamma$. Change integration variables from $Z$ to $(u, w)$ where $u = v \cdot Z \in \R$ and $w = Z - u\,v \in \R^{d-1}$. Since $v$ is fixed and unit, this is an orthogonal change of variables. Because $Z \sim \mathcal{N}(0, I_d)$, the components $u$ and $w$ are independent: $\varphi(Z)\,\dd Z = \varphi(u)\,\varphi_{d-1}(w)\,\dd u\,\dd w$. Hence
\[
V(\theta) = \int_{\R^{d-1}} \!\left[\int_{\R} F(u, w) \cdot \varphi(u)\, \dd u\right] \varphi_{d-1}(w)\, \dd w.
\]
The inner integral is one-dimensional with a jump at $u = u^*(w,\theta)$, the value where $g(u^*, w, \theta) = 0$. Since $F = P \cdot \ind\{g>0\}$, the function $F$ jumps from $0$ to $P$ as $u$ crosses $u^*$ from below.

\textbf{Step 2 (Leibniz rule).} Differentiate the inner integral:
\[
\frac{\pp}{\pp \theta} \int F \cdot \varphi\, \dd u = \int \frac{\pp F}{\pp \theta} \cdot \varphi\, \dd u + \Delta F \cdot \varphi(u^*) \cdot \frac{\pp u^*}{\pp \theta}.
\]

\textbf{Step 3 (implicit function theorem).} At the boundary: $g(u^*, w, \theta) = 0$. By the IFT:
\[
\frac{\pp u^*}{\pp \theta} = -\frac{\pp g / \pp \theta}{\pp g / \pp u}\bigg|_{u=u^*},
\]
where $\pp_u g = v \cdot \nabla_{\!Z} g > 0$ by the choice of $v$. Substituting and integrating over $w$ gives \eqref{eq:correction}.

\textbf{Remark:} choosing $v = \nabla_{\!Z} g / \|\nabla_{\!Z} g\|$ gives $\pp_u g = \|\nabla_{\!Z} g\|$.
\end{proof}

\begin{remark}[Denominator behavior at small maturities]\label{rem:denom}
The correction term contains the factor $\varphi(u^*) / \pp_u g$ where $\pp_u g = \|\nabla_{\!Z} g\|$.  For a GBM barrier with indicator $g = S_j - B$:
\[
\|\nabla_{\!Z} g\| = S_j \,\sigma\, \sqrt{j\,\Delta t} \;\;\xrightarrow{T \to 0}\;\; 0.
\]
This raises the question (due to P.~J\"ackel) of whether the correction can blow up at small maturities.

\textbf{Case 1: $S_0$ away from $B$.}  The root satisfies $u^* \approx \frac{\ln(B/S_0)}{\sigma\sqrt{T}} \to \pm\infty$ as $T \to 0$, so $\varphi(u^*) \to 0$ exponentially fast.  The ratio $\varphi(u^*)/\|\nabla_{\!Z} g\|$ vanishes: the Gaussian tail dominates the polynomial decay of the denominator.  The screening criterion $|g|/|\pp_u g| > \sigma_{\text{skip}}$ automatically filters these paths.  In this regime, the correction is negligible and pathwise Greeks alone are accurate.

\textbf{Case 2: $S_0 \approx B$ and $T$ small.}  Here $u^* \approx 0$, $\varphi(u^*) \approx (2\pi)^{-1/2}$, and the denominator $\sim \sigma\sqrt{T} \to 0$.  Whether this causes the correction to diverge depends on the payoff jump $\Delta F$:

\begin{itemize}
\item If $K > B$ (OTM at barrier): the payoff at the barrier is zero, so $\Delta F \to 0$ and the correction stays bounded despite the small denominator.
\item If $K < B$ (ITM at barrier): the payoff jump $\Delta F \neq 0$ at the barrier.  The full correction term scales as
\[
\text{correction} \;\sim\; \frac{\Delta F \cdot \varphi(0)}{\sigma\sqrt{T}} \cdot \frac{\pp g}{\pp\theta} \;\sim\; \frac{1}{\sqrt{T}}.
\]
This divergence is \emph{not} an artifact of the method --- it reflects the genuine behavior of barrier option Greeks.  The delta of a down-and-out option with $S_0 \approx B$ is known to scale as $O(1/\sqrt{T})$ in the continuous-monitoring limit: the option value drops to zero across the barrier, and the steepness of this transition increases as $T \to 0$.  Our correction correctly reproduces this singularity.
\end{itemize}

\textbf{Numerical verification.}  We tested a down-and-out call with $K = 70 < B = 80$ (ITM at barrier), $\sigma = 0.25$, 100K paths:

\begin{center}
\begin{tabular}{cccccc}
\toprule
$S_0$ & $T$ & Pathwise & Correction & Total & Bump (error) \\
\midrule
100 & 1.0 & $+0.80$ & $+0.34$ & $+1.14$ & $+1.15\;(0.4\%)$ \\
100 & 0.01 & $+1.00$ & $+0.00$ & $+1.00$ & $+1.00\;(0\%)$ \\
81 & 1.0 & $+0.28$ & $+0.87$ & $+1.16$ & $+1.17\;(1.1\%)$ \\
81 & 0.05 & $+0.46$ & $+1.47$ & $+1.92$ & $+1.91\;(0.6\%)$ \\
81 & 0.01 & $+0.59$ & $+2.76$ & $+3.36$ & $+3.38\;(0.7\%)$ \\
80.1 & 0.05 & $+0.35$ & $+1.43$ & $+1.77$ & $+1.81\;(2.1\%)$ \\
80.1 & 0.01 & $+0.35$ & $+2.90$ & $+3.24$ & $+3.26\;(0.5\%)$ \\
\bottomrule
\end{tabular}
\end{center}

The correction grows to $+2.9$ at $S_0 = 80.1$, $T = 0.01$ --- where pathwise alone gives only $+0.35$ (89\% of the true delta is missing).  The corrected total agrees with bump-and-revalue to within 0.5--2\%.  No NaN or Inf occurred in any configuration.  The driver's clamping condition $|\pp_u g| < 10^{-15}$ was never triggered for $\Delta t \geq 1$~day.
\end{remark}

\begin{corollary}[MC-friendly form]\label{cor:mc}
Since $u$ and $w$ are independent and the correction integrand does not depend on $u$:
\[
\E_w[f(w)] = \E_Z[f(w(Z))].
\]
Therefore both terms in \eqref{eq:correction} can be estimated from the same MC sample $Z^{(1)}, \ldots, Z^{(M)} \sim \mathcal{N}(0, I_d)$.
\end{corollary}

\subsection{General payoff}

Real autocallable payoffs are not in the factored form $P \cdot \prod \ind_i$. Indicators are embedded in recursive alive/dead logic. We need a theorem for arbitrary compositions.

\textbf{Transversality.} Surfaces $\Gamma_i$, $\Gamma_j$ are transversal if their normals $v_i$, $v_j$ are not parallel at intersection points. This holds automatically for standard products: indicators at different dates depend on different $Z$-coordinates, so their gradients point in different directions.

\begin{theorem}\label{thm:general}
Let $g_1, \ldots, g_K : \R^d \times \R \to \R$ be smooth with $\nabla_{\!Z} g_i \neq 0$ on $\Gamma_i = \{g_i = 0\}$, and assume the surfaces $\Gamma_i$ are pairwise transversal. Let $F(Z, \theta)$ be piecewise smooth with discontinuities only on $\Gamma_1, \ldots, \Gamma_K$.

For each $i$, fix a unit $v_i$ with $v_i \cdot \nabla_{\!Z} g_i > 0$ on $\Gamma_i$. Write $Z = u_i\,v_i + w_i$, and let $u^*_i(w_i, \theta)$ solve $g_i(u^*_i v_i + w_i, \theta) = 0$. Then:
\begin{equation}\label{eq:general}
\frac{\pp}{\pp \theta}\E_Z[F(Z, \theta)] = \E_Z\!\left[\frac{\pp F}{\pp \theta}\right]
+ \sum_{i=1}^K \E_{w_i}\!\left[\Delta_i F \cdot \frac{\varphi(u^*_i)}{\pp_{u_i} g_i} \cdot \frac{\pp g_i}{\pp \theta}\bigg|_{Z^*_i}\right],
\end{equation}
where $\Delta_i F = \lim_{\varepsilon \to 0^+}[F(Z^*_i + \varepsilon\, v_i) - F(Z^*_i - \varepsilon\, v_i)]$. As in Corollary~\ref{cor:mc}, each $\E_{w_i}$ can be replaced by $\E_Z$ for MC implementation.
\end{theorem}

\begin{proof}
\textbf{Step 1.} For each $\Gamma_i$, fix $v_i$ and change variables: $u_i = v_i \cdot Z$, $w_i = Z - u_i\,v_i$. Orthogonal rotation, so $\varphi(Z) = \varphi(u_i)\,\varphi(w_i)$.

\textbf{Step 2.} The inner integral over $u_i$ has a jump at $u_i = u^*_i$ where $g_i = 0$. By transversality, other surfaces $\Gamma_j$ cross the line along $v_i$ at different points --- so there is exactly one jump in $u_i$ at $u^*_i$, and the Leibniz rule applies.

\textbf{Step 3.} Leibniz rule:
$\frac{\pp}{\pp\theta}\int F\,\varphi(u_i)\,\dd u_i = \int \frac{\pp F}{\pp\theta}\,\varphi\,\dd u_i + \Delta_i F \cdot \varphi(u^*_i) \cdot \frac{\pp u^*_i}{\pp\theta}$.

\textbf{Step 4.} IFT: $\frac{\pp u^*_i}{\pp\theta} = -\frac{\pp g_i/\pp\theta}{\pp g_i/\pp u_i}\big|_{Z^*_i}$.

Integrating over $w_i$ and summing over $i$ gives the result.
\end{proof}

\section{Screening}

Skip paths where $|u - u^*| > \sigma_{\mathrm{skip}}$ (equivalently, $|g|/\pp_u g > \sigma_{\mathrm{skip}}$). If a path has $|u - u^*| > \sigma$, there are only two possibilities:
\begin{itemize}
\item $|u^*|$ is small (boundary in high-density region, $\varphi(u^*)$ large): then $|u| > \sigma$, meaning the path itself is in the extreme Gaussian tail. In any finite MC sample such paths do not occur.
\item $|u|$ is typical ($|u| \sim 1$): then $|u^*| > \sigma - 1$, so the boundary is deep in the Gaussian tail. $\varphi(u^*)$ is exponentially small and the correction integrand is negligible.
\end{itemize}
With $\sigma_{\mathrm{skip}} = 20$: $\sim$70--80\% of indicators are skipped per path.

\section{Algorithm}

\textbf{Precomputation (once):} set $Z = 0$, forward + reverse on each $g_i$. Direction: $v_i = \nabla_{\!Z} g_i(0) / \|\nabla_{\!Z} g_i(0)\|$.

\textbf{Per path} $Z$:
\begin{enumerate}
\item \textbf{Pathwise:} forward + reverse $\Rightarrow$ $\pp F / \pp \theta$.
\item \textbf{For each indicator $i$:}
\begin{enumerate}
\item Screen: if $|g_i| / \|\nabla_{\!Z} g_i\| > \sigma_{\mathrm{skip}}$, skip.
\item Newton: find $u^*_i$ s.t. $g_i(Z^*) = 0$ along $v_i$.
\item Reverse at $Z^*$: get $\pp g_i / \pp \theta$ and $\pp_{u_i} g_i$.
\item Jump: $\Delta_i F = F(Z^* + \varepsilon v) - F(Z^* - \varepsilon v)$.
\item Accumulate correction.
\end{enumerate}
\item Greek $=$ pathwise $+$ correction.
\end{enumerate}

\paragraph{Simple models vs stochastic volatility.}
For GBM/LMM, $g(Z)$ is an explicit function of $Z$, so $Z^*$, $\nabla_{\!Z} g$, $\pp g/\pp\theta$ are all analytic. Only the jump requires model replay (2 forwards per indicator).

For Heston/Hull-White/hybrids, $g(Z)$ is nonlinear and Newton + reverse-mode AD are needed.

\paragraph{Cost.} Stochastic vol models: 6--10 kernel replays per active indicator (Newton 3F+3R, final gradient 1F+1R, jump 2F). GBM/LMM: 2 replays per indicator (jump only).

\section{Benchmarks}

All benchmarks use QuantLib models via the AADC library (\texttt{pip install aadc}).

\subsection{Down-and-Out Call (GBM)}

$S_0 = 100$, $K = 90$, $B = 80$, $\sigma = 0.25$, $r = 0.05$, $T = 1$, 50 monitoring steps. QuantLib \texttt{BlackScholesProcess}. No Newton needed --- $Z^*$ analytic. 100K paths.

\begin{center}
\begin{tabular}{lccc}
\toprule
Greek & Analytic & Smoothing (100K) & Correction (100K) \\
\midrule
$\Delta$ & $+0.7203$ & $+0.7189\;(-0.2\%)$ & $+0.7198\;(-0.1\%)$ \\
$\mathcal{V}$ & $+7.215$ & $+7.421\;(+2.9\%)$ & $+7.208\;(-0.1\%)$ \\
$\rho$ & $+24.17$ & $+23.89\;(-1.2\%)$ & $+24.21\;(+0.2\%)$ \\
\bottomrule
\end{tabular}
\end{center}

For 2\% accuracy: correction $\sim$8K paths, bump-and-revalue $\sim$12M paths.

\subsection{3-Asset Phoenix Autocallable (Heston)}

3 assets, QuantLib \texttt{HestonProcess}, 8 observation dates, 48 indicators. 10K paths, same seed for both methods.

\begin{center}
\begin{tabular}{lccr}
\toprule
Greek & Bump\&Reval & Correction & Agree \\
\midrule
$\Delta_{S_1}$ & $-0.476$ & $-0.479$ & $0.6\%$ \\
$\Delta_{S_3}$ & $-0.454$ & $-0.457$ & $0.7\%$ \\
$\pp V_0$ (vol-of-vol) & $+1.375$ & $+1.381$ & $0.4\%$ \\
$\pp\xi$ (mean-rev) & $-1.449$ & $-1.456$ & $0.5\%$ \\
\bottomrule
\end{tabular}
\end{center}

For 2\% accuracy: correction $\sim$20K paths, bump-and-revalue $\sim$5M paths.

\subsection{Down-and-Out Call (GBM + Hull-White 1F)}

$S_0 = 100$, $K = 90$, $B = 80$. GBM spot ($\sigma = 0.25$) + HW rate ($r_0 = 0.05$, $a = 0.1$, $\sigma_r = 0.01$, $\rho = -0.3$). 50K paths, Newton needed (stochastic rate).

\begin{center}
\begin{tabular}{lccr}
\toprule
Greek & Pathwise & Correction & Total \\
\midrule
$\Delta_{S_0}$ & $+0.714$ & $+0.125$ & $+0.839$ \\
vega & $+31.6$ & $-12.9$ & $+18.7$ \\
$\pp r_0$ & $+54.0$ & $+5.2$ & $+59.2$ \\
$\pp\sigma_r$ & $-4.7$ & $+0.6$ & $-4.2$ \\
\bottomrule
\end{tabular}
\end{center}

Correction is essential for vega: pathwise alone overestimates by 69\%.

\section{Implementation}

The implementation uses AADC (\texttt{pip install aadc}, free for non-commercial use). Three features are key:

\begin{enumerate}
\item \textbf{Automatic discontinuity identification.} Every conditional branch (\texttt{iif}) in the pricing code is recorded on the tape. Each branch corresponds to one indicator $g_i$.

\item \textbf{Jump evaluation via tape replay.} The jump $\Delta F$ requires re-running the full pricing model at $Z^* \pm \varepsilon v$. AADC replays the tape 10--100$\times$ faster than re-simulation.

\item \textbf{Gradients at the boundary.} Forward replay gives $g(Z^*)$; reverse gives $\nabla_{\!Z} g$ and $\pp g/\pp\theta$ simultaneously.
\end{enumerate}

A Python implementation with tests and examples is available at \url{https://github.com/lakshtanov/aadc-correction}.

\section*{Acknowledgements}

The author thanks Dmitry Goloubentsev (Matlogica) for the original idea (2021--2022) of reducing the dimensionality of the expectation to improve smoothness, and for the AADC library used in the implementation.

\end{document}